\documentclass[11pt]{article}
\usepackage[left=1in, right=1in, top=1in, bottom=1in, margin=1in]{geometry}

\usepackage[utf8]{inputenc} 
\usepackage[T1]{fontenc}
\usepackage{verbatim}

\usepackage[backend=bibtex, style = numeric, doi = false, url = false, isbn = false, maxbibnames = 6]{biblatex}
\bibliography{references.bib}
\renewbibmacro{in:}{}      
\newbibmacro{string+doi}[1]{\iffieldundef{doi}{#1}{\href{https://dx.doi.org/\thefield{doi}}{#1}}}
\DeclareFieldFormat{title}{\usebibmacro{string+doi}{\mkbibemph{#1}}}
\DeclareFieldFormat[article]{title}{\usebibmacro{string+doi}{\mkbibquote{#1}}}
\DeclareFieldFormat[incollection]{title}{\usebibmacro{string+doi}{\mkbibquote{#1}}}                   
\DeclareFieldFormat[inproceedings]{title}{\usebibmacro{string+doi}{\mkbibquote{#1}}}     

\usepackage[colorlinks=true,linkcolor=blue,urlcolor=blue,citecolor=blue,anchorcolor=green,pdfusetitle]{hyperref}

\usepackage[cmex10]{amsmath}  % Use the [cmex10] option to ensure complicance
\usepackage{amsfonts}
\usepackage{amssymb,amsthm,mathtools,bbm}

\usepackage{cleveref}

\usepackage{mleftright}       % fix to wrong spacing of \left-,
\mleftright                   % \middle- \right-commands 

\usepackage{graphicx}         % provides \includegraphics{...} to
\usepackage{booktabs}         % fixes poor spacing in tables and
\newtheorem{thm}{Theorem}
\newtheorem{prop}[thm]{Proposition}
\newtheorem{lem}[thm]{Lemma}
\newtheorem{cor}[thm]{Corollary}
\theoremstyle{definition}

\newtheorem{ex}[thm]{Example}

\usepackage{braket}
\usepackage{mathtools}
\usepackage{dsfont}
\usepackage{xcolor}

\newcommand{\cE}{\mathcal{E}}

\newcommand{\cH}{\mathcal{H}}

\newcommand{\Z}{\mathbb{Z}}
\newcommand{\1}{\mathbbm{1}}

\DeclareMathOperator{\tr}{tr}
\DeclareMathOperator{\im}{im}

\DeclareMathOperator{\ord}{ord}
\DeclareMathOperator{\reg}{reg}
\DeclareMathOperator{\rank}{rank}

\newcommand{\C}{\mathbb{C}}

\newcommand{\wh}{\widehat}

\newcommand{\psucc}{p_{\mathrm{succ}}}

\definecolor{cool_green}{rgb}{0.0, 0.5, 0.0}

\usepackage[affil-it]{authblk}

\begin{document}
\title{On the sample complexity of the generalized hidden shift problem over arbitrary finite groups}
\author{Juntai Zhou\thanks{\href{mailto:juntaiz2@illinois.edu}{juntaiz2@illinois.edu}}}
\affil{\small Department of Mathematics, University of Illinois Urbana-Champaign}
\date{} 

\maketitle

\begin{abstract}
    We formulate the generalized hidden shift problem over arbitrary finite groups and obtain upper and lower bounds for the sample complexity in terms of group theoretic data.
\end{abstract}

\section{Introduction}

$\quad\:$The strongest known quantum algorithms often obtain speedups by recovering global algebraic structure hidden by an oracle. Simon's algorithm gives an exponential oracle-query advantage for finding a hidden xor-mask, while Shor's period-finding algorithm enables polynomial-time factoring and discrete logarithms \cite{Simon1997,Shor1997}. More generally, in hidden subgroup problem (HSP), given a finite group $G$ and an oracle $f:G\rightarrow S$ that is constant on each left coset of an unknown subgroup $H\leq G$ and distinct across cosets, the goal is to recover $H$. Abelian HSPs, including Simon's and Shor's problems, are efficiently solved using the quantum Fourier transform. For nonabelian groups, polynomial query complexity need not imply efficient measurements or post-processing, which keeps these variants central to quantum algorithms research \cite{EttingerHoyerKnill2004,HallgrenRussellTaShma2003,ChildsVanDam2010}.

Hidden shift problems encode a related symmetry. Given functions $f_0,f_1:G\rightarrow S$ satisfying
\begin{equation}
    f_0(x)=f_1(xs),
    \label{eq:hidden-shift}
\end{equation}
for an unknown $s\in G$, the task is to recover $s$ under suitable injectivity or distinctness promises. Structured cases such as shifts of the Legendre symbol are efficiently solvable \cite{VanDamHallgrenIp2006}; the injective cyclic case over $\Z_N$, however, is equivalent to the dihedral HSP. Kuperberg's sieve gives a subexponential-time quantum algorithm, but no accepted general polynomial-time solution is known \cite{Kuperberg2005}.

Childs and van Dam introduced a controlled interpolation between the difficult dihedral endpoint and an easy abelian endpoint \cite{ChildsVanDam2007}. Let $2\leq M\leq N$, let $S$ be a finite set, and suppose there is oracle access to
\begin{equation}
    f:\{0,1,\ldots,M-1\}\times \Z_N \longrightarrow S.
\end{equation}
For each fixed $b$, the slice $f_b(x):=f(b,x)$ is promised to be injective. In addition, there is an unknown $s\in\Z_N$ such that
\begin{equation}
    f(b,x)=f(b+1,x+s)
    \qquad
    \text{for } b=0,1,\ldots,M-2.
    \label{eq:generalized-promise}
\end{equation}
The \emph{generalized hidden shift problem} is to determine $s$. Equation~\eqref{eq:generalized-promise} says that the $M$ injective functions are successive translates of one another by the same hidden displacement. Equivalently, after fixing the first slice, one may write $f_b(x)=f_0(x-bs)$.

The parameter $M$ changes the algebraic structure available to the algorithm. When $M=2$, the promise reduces to the standard injective hidden shift problem over $\Z_N$, and hence to the dihedral HSP. When $M=N$, the first coordinate can be treated modulo $N$, and $f$ is constant on cosets of the cyclic subgroup
\begin{equation}
    \langle(1,s)\rangle \leq \Z_N\times\Z_N.
\end{equation}
This is an abelian HSP and is efficiently solvable by Fourier sampling. For intermediate values of $M$, the first coordinate does not wrap around to supply a full group action, and the problem is not naturally an HSP over a finite group. The family therefore isolates a meaningful transition: decreasing $M$ gradually removes the redundant shifted views that make the abelian endpoint easy while preserving a precise, analyzable promise.

The generalized hidden shift problem can be extended from cyclic groups to finite nonabelian groups. Childs and Wocjan explicitly identified the $M$-function nonabelian version as a natural generalization and showed that the two-function case over $S_n$ provides a direct formulation of rigid graph isomorphism and related group-action isomorphism problems, including code equivalence \cite{ChildsWocjan2007}. The nonabelian setting also exposes why additional shifted views may be valuable: on symmetric-group instances, single-register measurements require exponentially many hidden-shift states, while closely related hidden-subgroup results rule out strong Fourier sampling and point toward entangled measurements across multiple registers \cite{ChildsWocjan2007,MooreRussellSchulman2008}. Consequently, varying $M$ in the nonabelian problem could provide a controlled way to study how redundancy in the oracle, noncommutative representation theory, and collective measurement complexity interact, with potential implications for quantum algorithms for isomorphism testing and other symmetry-finding problems beyond the reach of abelian Fourier methods.

In this article, we study the sample complexity of the generalized hidden shift problem over generic finite groups. The remainder of this article is organized as follows: In section 2, we introduce preliminaries for the concepts and tools that will be used, including basic representation theory and quantum Fourier transform. In section 3, we rigorously formulate the generalized hidden shift problem over generic finite groups, which reduces to a standard state discrimination problem, and rewrite the states to be distinguished into block-diagonal form by applying quantum Fourier transform and Schur's orthogonality. In section 4, we analyze the sample complexity of this problem by applying generalized Holevo-Curlander inequality \cite{tyson2010two-sided,tyson2009two-sided}, obtaining upper and lower bounds in terms of group theoretic data. These results recover the analysis for cyclic groups by Childs and van Dam \cite{Childs2006generalized_hidden_shift}, while showing a nontrivial $\Omega(\log |G|)$ lower bound when $G=(\Z_2)^n$ and $\Omega(\sqrt{\log|G|})$ for $G=S_n$ (while Montanaro \cite{Montanaro2019PrettySimpleBounds} gives a trivial upper bound $O(\log|G|)$ for any finite group $G$).

\section{Preliminaries}
\subsection{Quantum state discrimination }
Let $\cE=(p_i,\rho_i)_{i=1}^N$ be a quantum state ensemble consisting of quantum states $\rho_1,\dots,\rho_N$ and a probability distribution $(p_1,\dots,p_N)$.
In the minimum average error setting, the optimal success probability $\psucc^*$ of distinguishing the states in $\cE$ is given by the following semidefinite program (SDP):
\begin{align}
	\psucc^* = \max\left\lbrace \sum\nolimits_{i=1}^N p_i \tr(M_i \rho_i) : M_i\geq 0 \text{ for all $i=1,\dots,N$}, \sum\nolimits_{i=1}^N M_i = \1\right\rbrace.
	\label{eq:primal}
\end{align}
The dual program gives the same value because of strong duality, and can be expressed as follows:
\begin{align}
	\psucc^* = \min\left\lbrace \tr K : K\geq p_i\rho_i \text{ for all $i=1,\dots,N$}\right\rbrace.
	\label{eq:dual}
\end{align}

A commonly considered measurement in state discrimination is the \emph{pretty good measurement }(PGM) or \emph{square-root measurement} \cite{belavkin1975optimal,holevo1979asymptotically,hausladen1994pretty}, defined as 
\begin{align}
	M_i=p_i\overline{\rho}^{-\frac{1}{2}}\rho_i\overline{\rho}^{-\frac{1}{2}} \label{eq:pgm}
\end{align} 
where $\overline{\rho}=\sum_i p_i\rho_i$ is the average state and the inverse is taken on the support of $\overline{\rho}$.
Without loss of generality, one can restrict the Hilbert space on which the states $\rho_i$ act to this support without changing the optimal value in \eqref{eq:primal} or \eqref{eq:dual}.

\subsection{Representation theory of finite groups}
We follow the treatment in \cite{FultonHarris1991RepTheory}. A unitary representation of a finite group $G$ is a pair $(\mu,\mathcal{H}_\mu)$ where $\mathcal{H}_\mu$ is a Hilbert space and $\mu:G\to U(\mathcal{H}_\mu)$ is a homomorphism. The character of a representation $\pi$ is the function $\chi_\pi(g):=\tr\pi(g)$. We denote $\hat{G}$ to be the set of irreducible representations of $G$ and $\mathbb{C}[G]\cong\bigoplus_{\pi\in\hat{G}}\cH_\pi^*\otimes\cH_\pi$ to be the group algebra. The regular representation acts on the group algebra by $\reg_G(g):=\bigoplus_{\pi\in\hat{G}}\1_{\cH_\pi^*}\otimes\pi(g)$ with character $\chi_{\reg}(g)=|G|\delta_{ge}$. We will be using the following standard results in representation theory:

\begin{lem}[Fourier basis of group algebra]\label{lemma: Fourier basis of group algebra}
    If $G$ is a finite abelian group, then the characters define an orthonormal basis for $\C[G]$:
    \begin{align}
        \ket{e_\chi}:=\frac{1}{\sqrt{|G|}}\sum_{g\in G}\overline{\chi(g)}\ket{h}
    \end{align}
    and the regular representation acts by
    \begin{align}
        \reg_G(g)\ket{e_\chi}=\chi(g)\ket{e_\chi}.
    \end{align}
\end{lem}

\begin{lem}[Schur's orthogonality]\label{lemma: Schur's orthogonality}
For any irreducible unitary representations $\mu,\nu\in\hat{G}$,
    \begin{align}
    \frac{1}{|G|}\sum_{g\in G}\mu(g)_{ij}\overline{\nu(g)_{st}}=\frac{1}{d_\mu}\delta_{\mu\nu}\delta_{is}\delta_{jt}\quad\forall\:1\leq i,j\leq\dim\cH_\mu,\:1\leq s,t\leq\dim\cH_\nu.
\end{align}
\end{lem}

\iffalse
\begin{lem}[column orthogonality of characters]\label{lemma: column orthogonality of characters}
    Denote $Cl(g)$ to be the conjugacy class of $g\in G$. Then
    \begin{align}
        \sum_{\mu\in\hat{G}}|\chi_\mu(g)|^2=\frac{|G|}{|Cl(g)|}.
    \end{align}
\end{lem}

\begin{lem}[projection formula]\label{lemma: projection formula}
    The operator $\frac{1}{|G|}\sum_{g\in G}\mu(g)$ projects onto the fixed point subspace $\mathcal{H}_\mu^G:=\{\ket{v}\in\mathcal{H}_\mu:\mu(g)\ket{v}=\ket{v}\:\forall\:g\in G\}$.
\end{lem}

\subsubsection{Regular representation}
\begin{lem}\label{lemma: regular representation character}
    \begin{align}
        \chi_{\reg}(x)=|G|\delta_{xe}.
    \end{align}
\end{lem}

\subsubsection{Eigenvalues of $\mu(g)$}
A standard lemma in linear algebra:
\begin{lem}\label{lemma: mu(g) diagonalizable}
    If $A\in\text{Mat}_{d\times d}(\mathbb{C})$ has order $n$, i.e. $A^n=\1$, then $A$ is diagonalizable over $\mathbb{C}$ with eigenvalues among the $n-$th roots of unity.
\end{lem}
\begin{proof}
    The minimal polynomial of $A$ splits because it divides the polynomial $x^n-1=\prod_{j\in\Z/n}(x-\omega_n^j)$ by Cayley-Hamilton theorem, so $A$ is diagonalizable with eigenvalues $\omega_n^j$ for $j$ in a subset of $\Z/n$.
\end{proof}
In particular, $\mu(g)$ is diagonalizable for any representation $\mu$ and any $g\in G$ with eigenvalues being roots of unity.
\fi

\subsection{Fourier transform over finite groups}
We follow the treatment in \cite{Moore2004GenericQFT}. Given $f:G\to\mathbb{C}$, the classical Fourier transform of $f$ is the function $\hat{f}:\hat{G}\to\mathbb{C}[G]\cong\bigoplus_{\mu\in\hat{G}}\text{End}(\mathcal{H}_\mu)$ given by
\begin{align}
    \hat{f}(\mu)=\sum_{g\in G}f(g)\mu(g).
\end{align}

The Quantum Fourier Transform (QFT) over a generic group $G$ is given by
\begin{align}
    F_G:\ket{g}\mapsto\frac{1}{\sqrt{|G|}}\sum_{\mu\in\hat{G}}\sqrt{d_\mu}\sum_{i,j=1}^{d_\mu}\mu(g)_{ij}\ket{\mu,i,j}.
\end{align}
By Schur's orthogonality one can verify that $F_G$ is a unitary. In particular, given function $f:G\to\mathbb{C}$, if we encode $\ket{f}:=\sum_g f(g)\ket{g}$, then 
\begin{align}
    F_G\ket{f}=\sum_{g\in G}f(g)F_G\ket{g}=\sum_{\mu\in\hat{G}}\sum_{i,j=1}^{d_\mu}\sqrt{\frac{d_\mu}{|G|}}(\sum_{g\in G}f(g)\mu(g)_{ij})\ket{\mu,i,j}=\sum_{\mu\in\hat{G}}\sum_{i,j=1}^{d_\mu}\sqrt{\frac{d_\mu}{|G|}}\hat{f}(\mu)_{ij}\ket{\mu,i,j},
\end{align}
therefore the quantum Fourier transform is exactly the unitary encoding the classical Fourier transform operation.

\section{Generalized hidden shift problem over generic groups}
\subsection{Set up}
Given $G$, an integer $m$, and a function $f:\{0,1,...,m-1\}\times G\to S$ satisfying:\\
(1) $f(i,\cdot)$ is injective for all $i\in\{0,1,...,m-1\}$;\\
(2) There exists $h\in G$ such that $f(i,g)=f(i+1,gh)$ for all $i\in\{0,1,...,m-2\}$ and $g\in G$.\\
The goal of the generalized hidden shift problem over $G$ is to identify the unknown shift element $h\in G$. This set up is a direct generalization of \cite{ChildsVanDam2007} where $G=\Z/n$ was considered.\\
To give a quantum algorithm, we first prepare the state
\begin{align}
    \frac{1}{\sqrt{m|G|}}\sum_{b\in\Z/m}\sum_{g\in G}\ket{b,g,f(b,g)}.
\end{align}
Measuring the third register and discarding its outcome produces
\begin{align}
    \rho_h=\frac{1}{|G|}\sum_{g\in G}\ket{\phi_{g,h}}\bra{\phi_{g,h}}
\end{align}
where
\begin{align}
    \ket{\phi_{g,h}}=\frac{1}{\sqrt{m}}\sum_{b\in\Z/m}\ket{b,gh^b}.
\end{align}
The problem now reduces to the state discrimination problem of the ensemble $\{\rho_h^{\otimes k}:h\in G\}$ with uniform prior distribution.

\subsection{Rewriting the hidden shift states into block-diagonal form}
Applying QFT to the second register of $\ket{\phi_{g,h}}$ we get
\begin{align}
    \ket{\tilde{\phi}_{g,h}}=\frac{1}{\sqrt{m|G|}}\sum_{b\in\Z/m}\sum_{\mu\in\hat{G}}\sqrt{d_\mu}\sum_{i,j=1}^{d_\mu}\mu(gh^b)_{ij}\ket{b}\ket{\mu,i,j}
\end{align}
and 
\begin{align}\label{eq: state before orthogonality}
    \tilde{\rho}_h&=\frac{1}{|G|}\sum_{g\in G}\ket{\tilde{\phi}_{g,h}}\bra{\tilde{\phi}_{g,h}}\\
    &=\frac{1}{m|G|^2}\sum_{g\in G}\sum_{b,c\in\Z/m}\sum_{\mu,\nu\in\hat{G}}\sqrt{d_\mu d_\nu}\sum_{i,j=1}^{d_\mu}\sum_{s,t=1}^{d_\nu}\mu(gh^b)_{ij}\overline{\nu(gh^c)_{st}}\ket{b,\mu,i,j}\bra{c,\nu,s,t}.
\end{align}
For any $\mu\in\hat{G}$, $1\leq i,j\leq d_\mu$, $b\in\Z/m$, because $\mu$ is a group homomorphism we have
\begin{align}
    \mu(gh^b)_{ij}=\sum_{l=1}^{d_\mu}\mu(g)_{il}\mu(h^b)_{lj}.
\end{align}
Applying Schur's orthogonality we get
\begin{align}
    \frac{1}{|G|}\sum_{g\in G}\mu(gh^b)_{ij}\overline{\nu(gh^c)_{st}}&=\sum_{l=1}^{d_\mu}\sum_{l'=1}^{d_\nu}\frac{1}{|G|}\mu(h^b)_{lj}\overline{\nu(h^c)_{l't}}\sum_{g\in G}\mu(g)_{il}\overline{\nu(g)_{sl'}}\\
    &=\frac{\delta_{\mu\nu}\delta_{is}}{d_\mu}\sum_{l=1}^{d_\mu}\mu(h^b)_{lj}\overline{\mu(h^c)_{lt}}\\
    &=\frac{\delta_{\mu\nu}\delta_{is}}{d_\mu}\mu(h^{b-c})_{tj}
\end{align}
and therefore \cref{eq: state before orthogonality} becomes (after adjusting the basis to drop the transpose on $\mu(h^{b-c})$)
\begin{align}
    \tilde{\rho}_h&=\frac{1}{m|G|}\bigoplus_{\mu\in\hat{G}}\sum_{b,c\in\Z/m}\ket{\mu,b}\bra{\mu,c}\otimes\sum_{i,j,t=1}^{d_\mu}\mu(h^{b-c})_{tj}\ket{i,j}\bra{i,t}\\
    &\cong\frac{1}{m|G|}\bigoplus_{\mu\in\hat{G}}\sum_{b,c\in\Z/m}\ket{\mu,b}\bra{\mu,c}\otimes\1_{d_\mu}\otimes\mu(h^{b-c})\\
    &\cong\frac{1}{|G|}\bigoplus_{\mu\in\hat{G}}\1_{d_\mu}\otimes\tau_\mu(h)
\end{align}
where 
\begin{align}
    \tau_{\mu}(h):=\frac{1}{m}\sum_{b,c\in\Z/m}\ket{b}\bra{c}\otimes\mu(h^{b-c}). 
\end{align}
We can further rewrite
\begin{align}
    \tilde{\rho}_h&=\frac{1}{m|G|}\sum_{b,c\in\Z/m}\ket{b}\bra{c}\otimes\reg_G(h^{b-c}).
\end{align}
Note that 
\begin{align}
    (\1\otimes\reg_G(g))\tilde{\rho}_h(\1\otimes\reg_G(g))^\dagger=\tilde{\rho}_{ghg^{-1}},
\end{align}
therefore $\{\tilde{\rho}_h:h\in G\}$ is a compound geometrically uniform (CGU) ensemble \cite{Eldar2004OptimalDetection, zhou2025geometrically}, with generators in one to one correspondence with the conjugacy classes of $G$. Finally,
\begin{align}
\tilde{\rho}_h^{\otimes k}
\;=\;
\bigoplus_{\vec{\mu}\in(\wh G)^k}
\left(\frac{1}{|G|^k}\bigotimes_{j=1}^k \1_{d_{\mu_j}}\right)\otimes\tau_{\vec{\mu}}(h)
\end{align}
and the PGM operators of the ensemble $\{\tilde{\rho}_h^{\otimes k}:h\in G\}$ are
\begin{align}
E_h=S^{-1/2}\Big(\tfrac{1}{|G|}\tilde{\rho}_h^{\otimes k}\Big)S^{-1/2}=\bigoplus_{\vec{\mu}}\big(\frac{1}{|G|}\bigotimes_{j=1}^k \1_{d_{\mu_j}}\big)\otimes
\Big(\overline{\tau}_{\vec{\mu}}^{-1/2}\,\tau_{\vec{\mu}}(h)\,\overline{\tau}_{\vec{\mu}}^{-1/2}\Big)
\end{align}
where the average state
\begin{align}
S:=\frac{1}{|G|}\sum_{h\in G}\tilde{\rho}_h^{\otimes k}=\bigoplus_{\vec{\mu}}(\frac{1}{|G|^k}\bigotimes_{j=1}^k\1_{d_{\mu_j}})\otimes\overline{\tau}_{\vec{\mu}}.
\end{align}

\section{Lower and upper bounds for sample complexity}
\subsection{Sample complexity of state discrimination}
Given ensemble $\cE=(p_i,\rho_i)_{i=1}^{n}$, consider the $k-$copy i.i.d. state ensemble $\cE_k:=(p_i,\rho_i^{\otimes k})_{i=1}^n$. The sample complexity of state discrimination given constant $\epsilon>0$ is defined to be
\begin{align}
    k_{\min}(\cE,\epsilon):=\min\{k:\psucc^*(\cE_k)\geq\epsilon\}.
\end{align}
Note that if we only care about the asymptotic growth of $k_{\min}$ in terms of $n$, then the choice of $\epsilon$ does not affect the order. In the rest of this article, we may drop the error threshold $\epsilon$ and simply consider
\begin{align}
    k_{\min}(\cE):=\min\{k:\psucc^*(\cE_k)\geq\frac{2}{3}\}.
\end{align}
We first review the upper bound of $k_{\min}$ given by Barnum-Knill's pairwise-fidelity bound \cite{BarnumKnill2002}:

\begin{lem}[Barnum-Knill \cite{BarnumKnill2002}]\label{lemma: Barnum-Knill}
    Given an ensemble $\mathcal{E}:=\{(p_i,\rho_i)\}$,
    \begin{align}
        1-p_{\text{PGM}}\leq\sum_{i\neq j}\sqrt{p_ip_j}F(\rho_i,\rho_j).
    \end{align}
\end{lem}

\begin{cor}[Montanaro \cite{Montanaro2019PrettySimpleBounds}]\label{cor: trivial upper bound of k_min}
    Suppose $p_i=\frac{1}{n}$. Let $\mu(\mathcal{E}):=\max_{i\neq j}F(\rho_i,\rho_j)<1$, then
    \begin{align}
        k_{\min}\leq O(\frac{\log n}{-\log\mu}).
    \end{align}
\end{cor}
\begin{proof}
    Since fidelity is multiplicative,
    \begin{align}
        1-p_k\leq\frac{1}{n}\sum_{i\neq j}F(\rho_i,\rho_j)^k\leq\frac{1}{n}\sum_{i\neq j}\mu^k=(n-1)\mu^k.
    \end{align}
    Therefore, with $n\mu^k\leq\frac{1}{2}$ it suffices to make $p_k\geq\frac{1}{2}$, i.e. $k_{\min}\leq O(\frac{\log n}{-\log\mu})$.
\end{proof}

However, pairwise overlaps can still be difficult to compute. As an alternative, we will use the Holevo-Curlander estimate which gives a sandwiched bound on the optimal success probability by a trace quantity, therefore simplifying the expression of sample complexity as well:
\begin{lem}[generalized Holevo-Curlander \cite{ogawa1999strong,tyson2010two-sided,tyson2009two-sided}]\label{lemma: generalized Holevo-Curlander}
    The optimal success probability $\psucc^*$ of a \emph{generic} quantum state ensemble $(p_i,\rho_i)_{i\in [n]}$ satisfies
\begin{align}
    \left[ \tr\sqrt{ \sum\nolimits_{i=1}^n p_i^2 \rho_i^2} \right]^2 \leq \psucc^* \leq \tr\sqrt{\sum\nolimits_{i=1}^n p_i^2 \rho_i^2}.
    \label{eq:holevo-curlander}
\end{align}
\end{lem}

In particular, the asymptotic growth of $k_{\min}$ is of the same order as the minimum $k$ such that $q_k:=\tr\sqrt{\sum_i p_i^2(\rho_i^{\otimes k})^2}\geq\epsilon$. However, the operator $\rho_i^2$ still involves complicated spectral information in general, so we consider the following simple case where each $\rho_i$ is a normalized projector:

\begin{prop}\label{prop: uniform projectors}
    Suppose $p_i=\frac{1}{n}$ and $\rho_i=\frac{1}{r}P_i$ where $P_i=W_iW_i^\dagger$ is projection with rank $r$, i.e. $W_i$ has $r$ orthonormal columns. Let $B_k:=\begin{bmatrix}
        (\frac{1}{r}W_1)^{\otimes k} & ... & (\frac{1}{r}W_n)^{\otimes k}
    \end{bmatrix}$ and $G_k:=B_k^\dagger B_k$, then
    \begin{align}
        q_k:=\tr\sqrt{\sum_ip_i^2(\rho_i^2)^{\otimes k}}=\frac{1}{n}||B_k||_1=\frac{1}{n}\tr\sqrt{G_k}
    \end{align}
    where $||X||_1:=\tr\sqrt{X^\dagger X}=\tr\sqrt{XX^\dagger}$ is the nuclear norm.
\end{prop}
\begin{proof}
    Since $\rho_i^2=\frac{1}{r^2}P_i=(\frac{1}{r}W_i)(\frac{1}{r}W_i^\dagger)$, we have
    \begin{align}
        q_k=\frac{1}{n}\tr\sqrt{\sum_i(\frac{1}{r}W_i)^{\otimes k}(\frac{1}{r}W_i^\dagger)^{\otimes k}}=\frac{1}{n}\tr\sqrt{B_kB_k^\dagger}=\frac{1}{n}||B_k||_1.
    \end{align}
\end{proof}

A simple reformulation of Barnum-Knill estimate under this setup:
\begin{prop} \label{prop: Holevo-Curlander formulation of Barnum-Knill}
    \begin{align}\label{eq: Barnum-Knill estimate}
        q_k\geq\frac{1}{\sqrt{1+\frac{1}{n}\sum_{i\neq j}(\frac{\tr(P_iP_j)}{r})^k}}.    
    \end{align}
    In particular, if $\mu:=\max_{i\neq j}\frac{\tr(P_iP_j)}{r}<1$, then
    \begin{align}
        k_{\min}\leq O(-\frac{\log n}{\log\mu}).
    \end{align}
\end{prop}
\begin{proof}
    Since
    \begin{align}\label{eq: interpolation of Schatten norms}
        ||X||_1\geq\frac{||X||_2^3}{||X||_4^2},
    \end{align}
    we have
    \begin{align}
        q_k=\frac{1}{n}||B_k||_1\geq\frac{1}{n}\frac{||B_k||_2^3}{||B_k||_4^2}=\frac{1}{n}\frac{(n/r^k)^\frac{3}{2}}{(\sum_{i,j}\tr(P_iP_j)^k/r^{4k})^\frac{1}{2}}=\frac{1}{\sqrt{1+\frac{1}{n}\sum_{i\neq j}(\frac{\tr(P_iP_j)}{r})^k}}.
    \end{align}
\end{proof}

\iffalse
For convenience, we can further replace the Schatten-$\frac{1}{2}$ norm by the 1-norm:
\begin{lem}\label{lemma: kmin by 1-norm}
    Let
    \begin{align}
        s_k:=\frac{r^k||X_k-\frac{1}{r^{2k}}\1||_1}{2n}, 
    \end{align}
    then 
    \begin{align}
        1-s_k\leq\frac{1}{n}\tr\sqrt{X_k}\leq\sqrt{1-s_k^2}.
    \end{align}
    As a result,
    \begin{align}
        (1-s_k)^2\leq p_k^*\leq\sqrt{1-s_k^2}
    \end{align}
    and the sample complexity problem reduces to estimating
    \begin{align}
        k_{\min}:=\min\{k:s_k\leq\epsilon\}.
    \end{align}
\end{lem}
\begin{proof}
    Let $\rho_k:=\frac{r^k}{n}X$ and $\sigma_k:=\frac{1}{nr^k}\1$, then $F(\rho_k,\sigma_k)=\frac{1}{n}\tr\sqrt{X_k}$ and $||\rho_k-\sigma_k||_1=\frac{r^k}{n}||Y_k||_1=s_k$. By Fuchs–van de Graaf inequalities
    \begin{align}
        1-F(\rho_k,\sigma_k)\leq\frac{1}{2}||\rho_k-\sigma_k||_1\leq\sqrt{1-F(\rho_k,\sigma_k)^2}
    \end{align}
    we get
    \begin{align}
        1-s_k\leq\frac{1}{n}\tr\sqrt{X_k}\leq\sqrt{1-s_k^2}.
    \end{align}
\end{proof}
\fi

\subsection{Sample complexity of discriminating the generalized hidden shift states}
For simplicity, we consider the case where $m=|G|$.

\subsubsection{Upper bound}
Take a purification of $\tilde{\rho}_h$
\begin{align}
    \ket{\Phi_h}:=\frac{1}{\sqrt{m|G|}}\sum_{b\in\Z/m}\ket{b}\otimes\reg_G(h^b)\ket{\Omega},
\end{align}
then by Uhlmann's theorem
\begin{align}
    F(\tilde{\rho}_g,\tilde{\rho}_h)&=\max_{U_R}\bra{\Phi_g}(\1\otimes U_R)\ket{\Phi_h}\\
    &=\frac{1}{m|G|}\max_{U_R}\tr[U_R\sum_{b\in\Z/m}\reg_G(h^bg^{-b})]\\
    &=\frac{1}{m|G|}||\sum_{b\in\Z/m}\reg_G(h^bg^{-b})||_1.
\end{align}

\begin{prop}
    Define 
    \begin{align}
        r(g,h):=\frac{|\{b\in\Z/|G|:h^b=g^b\}|}{|G|}=\frac{1}{\min\{b>0:g^b=h^b\}}.
    \end{align}
    Then
    \begin{align}
        r(g,h)\leq F(\tilde{\rho}_g,\tilde{\rho}_h)\leq\sqrt{r(g,h)}.
    \end{align}
\end{prop}
\begin{proof}
    Denote $\reg_G$ to be the right regular representation of $G$ and $X:=\sum_{b\in\Z/|G|}\reg_G(h^bg^{-b})$, then
    \begin{align}
        \sqrt{F}(\tilde{\rho}_g,\tilde{\rho}_h)=\frac{1}{|G|^2}||\sum_{b\in\Z/|G|}\reg_G(h^bg^{-b})||_1=\frac{||X||_1}{|G|^2}.
    \end{align}
    Then
    \begin{align}
        ||X||_F^2&=\sum_{b,c\in\Z/|G|}\tr[\reg_G(h^bg^{-b})\reg_G(h^cg^{-c})^\dagger]\\
        &=|G||\{(b,c)\in(\Z/|G|)^2:h^{b-c}=g^{b-c}\}|\\
        &=|G|^3r(g,h),
    \end{align}
    so 
    \begin{align}
        |G|^2r(g,h)=|\tr(X)|\leq||X||_1\leq\sqrt{|G|}||X||_F=|G|^2\sqrt{r(g,h)},
    \end{align}
    therefore
    \begin{align}
        r(g,h)\leq F(\tilde{\rho}_g,\tilde{\rho}_h)\leq\sqrt{r(g,h)}.
    \end{align}
\end{proof}

As a simple consequence of \cref{lemma: Barnum-Knill}, we obtain

\begin{cor}\label{cor: Barnum-Knill estimate}
    \begin{align}
         1-p_{\text{PGM}}\leq\frac{1}{|G|}\sum_{g\neq h}r(g,h)^\frac{k}{2}.
    \end{align}
    Therefore
    \begin{align}
        k_{\min}\leq\min\{k:\frac{1}{|G|}\sum_{g\neq h\in G}r(g,h)^k\leq\epsilon\}.
    \end{align}
\end{cor}

The proof above directly computes the pairwise fidelities and applies Barnum-Knill. As suggested by \cref{prop: Holevo-Curlander formulation of Barnum-Knill}, the generalized Holevo-Curlander inequality can give a simpler proof: Note that $\tilde{\rho}_h=\frac{1}{|G|}W_hW_h^\dagger$ where $W_h:=\frac{1}{\sqrt{m}}\sum_{b\in\Z/m}\ket{b}\otimes\reg_G(h)^b$, and $W_h^\dagger W_h=\1$, therefore 
\begin{align}
    (\tilde{\rho}_h)^2=\frac{1}{|G|^2}W_hW_h^\dagger W_hW_h^\dagger=\frac{1}{|G|}\tilde{\rho}_h.
\end{align}
By \cref{lemma: generalized Holevo-Curlander} and \cref{prop: uniform projectors}, it suffices to estimate the minimum $k$ such that
\begin{align}
    q_k:=\frac{1}{|G|}\tr\sqrt{M_k}\geq\epsilon
\end{align}
where 
\begin{align}
    M_k:=\sum_{h\in G}(\frac{1}{|G|}\tilde{\rho}_h)^{\otimes k}.
\end{align}

    Since $\tilde{\rho}_h=\frac{1}{|G|}W_hW_h^\dagger$ where $\rank(W_h)=\frac{1}{|G|}$, 
    \begin{align}
        \frac{\tr(W_gW_g^\dagger W_hW_h^\dagger)}{|G|}=\frac{|\{b\in\Z/|G|:h^b=g^b\}|}{|G|}=r(g,h).
    \end{align}
    By \cref{eq: Barnum-Knill estimate},
    \begin{align}
        k_{\min}=\min\{k:q_k\geq\epsilon\}\leq\min\{k:\frac{1}{|G|}\sum_{g\neq h\in G}r(g,h)^k\leq\epsilon'\}.
    \end{align}

\iffalse
\begin{ex}
    When $G=(\Z/2)^n$, $r(g,h)=\frac{1}{2}$ for all $g\neq h$, so this bound only gives $k_{\min}=O(n)=O(\log|G|)$.
\end{ex}

\begin{ex}
    Consider only taking $t=3$ in the sum, we have
    \begin{align}
        C_3\geq\frac{1}{|G|}|\{x\in G:\ord(x)=3\}|^2-1.
    \end{align}
    In particular, when $G=S_n$ and $t=p$ is prime \cite{MoserWyman1955, MoserWyman1956, Wilf1986}, 
    \begin{align}
        |\{x\in G:\ord(x)=p\}|\sim\frac{1}{\sqrt{p}}n^{(1-\frac{1}{p})n}\exp(n^{\frac{1}{p}}-n^{1-\frac{1}{p}}),
    \end{align}
    which already gives $\min\{k:\frac{C_3}{3^k}\leq\epsilon\}\geq\Omega(n\log n)$, so this bound only gives $k_{\min}=O(n\log n)=O(\log|G|)$ for $G=S_n$.
\end{ex}
\fi

\subsubsection{Lower bound}
In this subsection denote $m=|G|$. Let $\ket{f_t}$ be the Fourier basis over $\Z/m$ such that $S_c\ket{f_t}=\omega_m^{tc}\ket{f_t}$ where $S_c:=\sum_{b\in\Z/m}\ket{b}\bra{b-c}$, then 
\begin{align}\label{eq: rho_h after DFT}
    \tilde{\rho}_h\cong\frac{1}{m}\bigoplus_{t\in\Z/m}\ket{f_t}\bra{f_t}\otimes F_t(h)\cong\frac{1}{m}\bigoplus_{t\in\Z/m}F_t(h)
\end{align}
where 
\begin{align}
    F_t(h):=\frac{1}{|G|}\sum_{b\in\Z/m}\omega_m^{tb}\reg_G(h)^b=\frac{1}{|G|}\sum_{r\in\Z/x_h}(\sum_{s\in\Z/y_h}(\omega_m^{tx_h})^s)\omega_m^{tr}\reg_G(h)^r=\begin{cases}
        0 & y_h\nmid t\\ E_t(h) & y_h\mid t
    \end{cases}
\end{align}
where $x_h:=\ord(h)$ and $y_h:=\frac{|G|}{x_h}$, and 
\begin{align}
    E_t(h):=\frac{1}{x_h}\sum_{r\in\Z/x_h}\omega_m^{tr}\reg_G(h)^r
\end{align}
is a projection for $y_h|t$ with rank
\begin{align}
    \tr E_t(h)=\frac{1}{x_h}\sum_{r\in\Z/x_h}\omega_m^{tr}\chi_{\reg}(h^r)=y_h.
\end{align}
Given $\vec{t}\in(\Z/m)^k$ we have
\begin{align}
    \tr(M_k(\vec{t}))&=\frac{1}{|G|^{2k}}\sum_{h\in G}\prod_{i=1}^{k}\tr F_{t_i}(h)\\
    &=\frac{1}{|G|^{2k}}\sum_{h\in G}(\frac{|G|}{\ord(h)})^k\:\1_{y_h\mid t_1, ..., y_h\mid t_k}\\
    &=\frac{1}{|G|^k}\sum_{s\in\Z/|G|}\frac{N_s}{s^k}\:\1_{\frac{|G|}{s}\mid t_1, ..., \frac{|G|}{s}\mid t_k}
\end{align}
where
\begin{align}
    N_s:=|\{h\in G:\ord(h)=s\}|.
\end{align}

\begin{prop}
    Let $R_G$ be the set of orders that occur in $G$ and $g(G):=\max R_G$; Note that $g$ is Landau's function when $G=S_n$. Then
    \begin{align}\label{eq: upper bound for q_k}
        q_k\leq\sqrt{\frac{g(G)^{k+1}}{|G|}}.
    \end{align}
\end{prop}
\begin{proof}
    By \cref{eq: rho_h after DFT}, we have
    \begin{align}
        \rank(M_k)\leq N_k|G|^k
    \end{align}
    where 
    \begin{align}
        N_k:=|\bigcup_{r\in R_G}\{\vec{t}\in(\Z/|G|)^k:\frac{|G|}{r}\mid t_i\:\forall\:i\}|\leq\sum_{r\in R_G}|\{\vec{t}:\frac{|G|}{r}\mid t_i\:\forall\:i\}|=\sum_{r\in R_G}r^k\leq|R_G|g(G)^k\leq g(G)^{k+1}.
    \end{align}
    Therefore
    \begin{align}
        q_k=\frac{1}{|G|}\tr\sqrt{M_k}\leq\frac{1}{|G|}\sqrt{\rank(M_k)\tr(M_k)}\leq\sqrt{\frac{g(G)^{k+1}}{|G|}}.
    \end{align}
\end{proof}

\begin{ex}
    When $G=(\Z/2)^n$, $g(G)=2$, so $k_{\min}=\Theta(\log|G|)$.
\end{ex}

\begin{cor}
    When $G=S_n$, we have the following lower bound for sample complexity:
    \begin{align}
        k_{\min}\geq\Omega(\sqrt{n\log n})=\Omega(\sqrt{\log|G|}).
    \end{align}
\end{cor}
\begin{proof}
    Using the asymptotic estimate of Landau's function \cite{Landau1903}
    \begin{align}
        \lim_{n\to\infty}\frac{\log g(S_n)}{\sqrt{n\log n}}=1,
    \end{align}
    \cref{eq: upper bound for q_k} shows that \begin{align}
        k_{\min}\geq\Omega(\frac{\log|G|}{\log g(G)})=\Omega(\sqrt{n\log n}).
    \end{align}
\end{proof}

\section{Conclusion}
We have formulated the generalized hidden shift problem over generic finite
groups and analyzed the sample complexity of distinguishing the associated
hidden-shift states under a uniform prior. The quantum Fourier transform and
Schur's orthogonality provide a block-diagonal description of the states and
the pretty good measurement. Exploiting the
fact that the hidden-shift states are normalized projectors, we have applied
the generalized Holevo--Curlander inequality to obtain upper and lower bounds
in terms of group-theoretic data. These estimates determine the tight scaling
$k_{\min}=\Theta(\log |G|)$ for $G=(\mathbb{Z}/2)^n$; and a lower bound $k_{\min}=\Omega(\sqrt{n\log n})=\Omega(\sqrt{\log |G|})$ for symmetric group $G=S_n$, where the generalized hidden shift problem provides a direct formulation of rigid graph isomorphism problem.

\section*{Acknowledgement and Disclosure of AI usage}
\paragraph*{Acknowledgement} The author acknowledges helpful discussions with Srinivasan Arunachalam and Felix Leditzky, and thanks Srinivasan Arunachalam for his mentorship during the externship program at IBM Research, where he suggested the idea that led to this work. This work was supported by a grant through the IBM-Illinois Discovery Accelerator Institute as well as National Science Foundation Grant No. 2426103. 

\paragraph*{Disclosure of AI usage}
The main idea and overall structure of this article were established in 2025. The literature search and the refinement and polishing of the proofs have been assisted by GPT-O3 and subsequent models. The author has reviewed and verified the entire article and assumes full responsibility for its accuracy and correctness.

\printbibliography
\end{document}